%% file: ifacconf.tex
\documentclass{ifacconf}

\usepackage{graphicx}      
\usepackage{natbib}        
\usepackage{amsmath,amssymb,mathtools}
\usepackage{subfig}
\usepackage{color}
\usepackage{enumitem}
\usepackage{amsmath}
\usepackage{url}
\input{notation}

\newtheorem{assumption}{Assumption}
\begin{document}
\begin{frontmatter}

\title{Probabilistic Recursively Feasible Motion Planning
Under Uncertain Environments \thanksref{footnoteinfo}} 

\thanks[footnoteinfo]{This work was supported by the IITP (Institute of Information \& Communications Technology Planning \& Evaluation)-ITRC(Information Technology Research Center) grant funded by the Korea government (Ministry of Science and ICT) (IITP-2025-RS-2023-00259991.}

\author[First]{Hyeontae Sung},
\author[First]{Hyeongchan Ham},
\author[First]{Junyoung Park},
\author[Second]{Kai Ren},
and
\author[First]{Heejin Ahn} 

\address[First]{School of Electrical Engineering, Korea Advanced Institute of Science and Technology, Daejeon, Korea (e-mail: \{hyeontae.sung, hyeongchan.ham, junyoung766, heejin.ahn\}@kaist.ac.kr)}
\address[Second]{The SYCAMORE Lab, École Polytechnique Fédérale de Lausanne (EPFL), Switzerland (e-mail: \{kai.ren\}@epfl.ch)}

\begin{abstract}
Safe motion planning in uncertain, time-varying environments is challenging because the safe region can change unpredictably across planning steps, often causing a loss of recursive feasibility. In this work, we present a Probabilistic Recursively Feasible Model Predictive Control (PRF-MPC) framework that guarantees recursive feasibility with a specified probability. We introduce properties that an ideal predictor should satisfy to ensure distributional consistency, and use these properties to derive closed-form expressions for the means and covariances of trajectories predicted at future time steps. Building on this analysis, we construct safety constraints that ensure, with high probability, that the current safe set is contained within the safe sets at future time steps, thereby probabilistically guaranteeing recursive feasibility. Simulation results on a lane-change scenario demonstrate that the proposed method significantly improves recursive feasibility.
\end{abstract}


\end{frontmatter}

\section{Introduction}
Safe motion planning is a fundamental task in safety-critical applications.
A safe motion planner must generate collision-free ego trajectories in uncertain and time-varying environments, such as moving obstacles.
To address this challenge, Model Predictive Control (MPC) is commonly utilized because it enables the safety guarantee through safety constraint formulations \citep{ahn2021safe, ren2022chance, lindemann2023safe}. 

Although these approaches perform well in many settings, they do not guarantee recursive feasibility.
Recursive feasibility ensures that once the MPC problem is feasible at a given planning step, it remains feasible at all subsequent steps. 
When recursive feasibility does not hold, the planner may violate its safety guarantee during closed-loop execution. In practice, maintaining recursive feasibility is challenging under uncertain and time-varying environments because the feasible region continually shifts over time. Consequently, a planned trajectory that is safe at the current step may become infeasible at future steps.

To address this problem, various approaches have been proposed to guarantee recursive feasibility. \cite{brudigam2021stochastic} utilizes fail-safe trajectory planning, and \cite{chen2023invariant} formulates a forward reachable set for safe motion planning. Both approaches consider formulating a control invariant set to ensure recursive feasibility. However, they are considered too conservative because they consider the worst-case reachable set of other vehicles. To overcome this, \cite{yang2026safe} proposes a contingency trajectory optimization framework based on forward reachable sets constructed from control-intent sets. However, they assume that prediction updates do not occur during the planning horizon, leaving the system vulnerable to sudden infeasibility when an update actually occurs. \cite{stamouli2024recursively} maintains recursive feasibility by merging feasible sets from previous planning steps, but this union can become overly permissive and compromise safety when predictions shift significantly.

In our previous work, \cite{ren2024recursively} guarantees recursive feasibility by assuming that the predicted distributions of obstacles do not shift significantly and their uncertainty shrinks sufficiently over time. However, our analysis shows that this assumption rarely holds in practice and is satisfied only with low probability, even under a known distribution of future trajectories. Although this assumption is a sufficient condition for recursive feasibility, the violation of the assumption degrades the recursive feasibility of MPC in practice. These limitations motivate alternative methods that can maintain recursive feasibility.

To guarantee recursive feasibility, we should consider the distribution propagation and model it explicitly. 
However, modeling the temporal propagation of predicted distributions is challenging with off-the-shelf prediction models \citep{ivanovic2019trajectron, varadarajan2022multipath++}, which typically minimize per-step prediction error without maintaining temporal stochastic consistency. This can result in temporally inconsistent distributions, making their closed-loop propagation difficult to characterize.
This behavior is also undesirable from a planning perspective, as such instability may induce oscillatory planning behaviors and compromise recursive feasibility even when prediction errors remain similar. This motivates us to propose key properties of an \textit{ideal predictor} that possess prediction-time invariance, enabling explicit modeling of distribution propagation.

We propose a Probabilistic Recursively Feasible MPC (PRF-MPC) framework that guarantees recursive feasibility with a specified probability in shrinking-horizon planning. Our approach constructs a safe set with probabilistic margins that account for future distributions at a chosen confidence level, thereby providing robustness to distribution shifts over future planning steps. To formulate these margins, we propose properties of an ideal predictor under a Gaussian assumption.

Our main contributions are summarized as follows.
\begin{itemize}
    \item  We design a recursively feasible MPC framework under uncertain dynamic obstacles with probabilistic guarantees.
    \item We propose properties of an \emph{ideal predictor}, and then formalize how the prediction distributions propagate over time.    
    \item We validate the proposed approach through simulation on an autonomous driving lane-changing scenario.
\end{itemize}

\textit{Notation:} We denote a sequence of consecutive integers from $a$ to $b$, i.e., $\{a, a+1, \dots, b \}$, as $\mathbb{Z}_{a:b}$. Let $\mathbb{P}(\cdot)$ refer to a probability measure. The Gaussian distribution is denoted as $\mathcal{N}(\mu, \Sigma)$, where $\mu$ is the mean and $\Sigma$ is the covariance matrix. We notate conjunction as $\bigwedge$ and disjunction as $\bigvee$. 
We use capital letters, such as $O$, to denote random variables.

\section{Problem Statement}

\subsection{Recursively Feasible Safe Motion Planning}\label{subsection: PS 1}
We consider an MPC problem that plans an optimal trajectory of the ego vehicle (EV) while guaranteeing collision avoidance from other vehicles (OVs) over a finite horizon. The EV follows a deterministic discrete-time dynamics 
\begin{equation}
  x_{t+1} = f(x_t, u_t), \qquad x_t \in \mathbb{R}^{n_x},\; u_t \in \mathbb{R}^{n_u}, \label{eq:systemdynamic}\nonumber
\end{equation}
subject to state and input constraints
\begin{equation}
  x_t \in \mathcal{X}\subset \mathbb{R}^{n_x}, \quad u_t \in \mathcal{U}\subset \mathbb{R}^{n_u}, \label{eq:xuconstraints}\nonumber
\end{equation}
where $\mathcal{X}$ and $\mathcal{U}$ are time-invariant convex sets. 

Let $T\in \mathbb{N}$ be the planning horizon, and $\tau$ be the planning step. We denote by $\mathcal{X}^{\text{safe}}_{t|\tau}$ the safe set at timestep $t$ that is constructed during planning step $\tau$. Since this set is defined in terms of an environment uncertainty, we impose it as a chance constraint with a risk tolerance $\epsilon$,
\begin{equation}
    \mathbb{P}\Bigg(\bigwedge_{t=\tau+1}^T x_t \in \mathcal{X}^{\text{safe}}_{t|\tau}\Bigg) \geq 1-\epsilon.\label{eq: safeset}
\end{equation}
We conservatively reformulate this joint chance constraint into individual chance constraints via Boole’s inequality,
\begin{equation}
    \forall t\in\mathbb{Z}_{\tau+1:T},\quad \mathbb{P}\Big( x_t \in \mathcal{X}^{\text{safe}}_{t|\tau}\Big) \geq 1-\epsilon_t,
    \label{eq: single safeset}
\end{equation}
where $\epsilon_t$ is the risk tolerance at timestep $t$ for planning step $\tau$, allocated from the overall risk tolerance $\epsilon$ for the planning horizon $T$, that is, $\overset{T}{\underset{t=1}{\sum}}\epsilon_t = \epsilon$. 

The shrinking-horizon MPC problem at planning step $\tau$ is
\begin{subequations} \label{eq: MPC}
    \begin{alignat} {3}
        & \underset{u_{\tau:T-1}}{\text{min}} & \quad &\mathcal{J}(x_{\tau+1:T}, u_{\tau:T-1})
        &
        \\
        & s.t. & \quad & x_{t+1} = f(x_t,u_t), \quad\quad\quad\;\; 
        &\forall t \in \mathbb{Z}_{\tau:T-1}, 
        \label{eq: dynamics}
        \\
        & & \quad &   x_t \in \mathcal{X}, \quad\quad\quad\quad\quad\quad\quad\; 
        &\forall t \in \mathbb{Z}_{\tau+1:T},
        \label{eq: x bounds}
        \\
        & & \quad &  u_t \in \mathcal{U}, \quad\quad\quad\quad\quad\quad\quad\quad 
        &\forall t \in \mathbb{Z}_{\tau:T-1}, 
        \label{eq: u bounds}
        \\
        & & \quad &
        \mathbb{P}\Big( x_t \in \mathcal{X}^{\text{safe}}_{t|\tau}\Big) \geq 1-\epsilon_t, \;
        &\forall t \in \mathbb{Z}_{\tau+1:T}, 
        \label{eq: single safe constraint}
    \end{alignat}
\end{subequations}
where $x_{\tau+1:T}=(x_{\tau+1}, \ldots, x_T)$ and $u_{\tau:T-1}=(u_{\tau},\ldots, u_{T-1})$. 
Following the setting of \cite{ren2024recursively}, we adopt the shrinking-horizon setting which is suitable for maneuvers requiring interaction with OVs such as a turn or lane-changing within a finite horizon.

We define the probabilistically safe set at timestep $t$ and planning step $\tau$ as the set of states for which the safety constraint at $t$ is satisfied with probability at least $1-\epsilon_t$,
\begin{equation}    
    \bar{\mathcal{X}}^{\text{ safe}}_{t|\tau} :=\{x_t\in \mathbb{R}^{n_x} \mid \mathbb{P}\big(x_t\in \mathcal{X}^{\text{safe}}_{t|\tau} \big) \geq 1-\epsilon_t\}.
    \label{eq: feasible set}
\end{equation}

Recursive feasibility requires that feasibility at one planning step $\tau$ implies feasibility at the next $\tau+1$. The following lemma gives a sufficient condition for this property in the shrinking-horizon MPC.
\begin{lem}\label{lem: RecurFeasi}
The MPC problem \eqref{eq: MPC} is recursively feasible if the safe sets satisfy the step-to-step inclusion,
\begin{equation}
\forall \tau \in \mathbb{Z}_{0:t-1}, \,\text{it holds that }
\bar{\mathcal{X}}^{\text{safe}}_{t|\tau}\subseteq \bar{\mathcal{X}}^{\text{safe}}_{t|\tau+1}, 
\, \forall\, t \in \mathbb{Z}_{2:T}.
\label{eq: set inclusion}
\end{equation}
\end{lem}

\begin{pf}
If \eqref{eq: MPC} is feasible at $\tau=k$, then we have $x_{t|k}\in \bar{\mathcal{X}}^{\text{safe}}_{t|k}$ for all $t\in \mathbb{Z}_{k+1:T}$. By \eqref{eq: set inclusion}, these states remain inside the safe sets at the next planning step, i.e., $x_{t|k} \in \bar{\mathcal{X}}^{\text{safe}}_{t|k+1}$ for all $t\in \mathbb{Z}_{k+2:T}$. Induction over $\tau$ yields the claim.~~~~~~~~~~~~~~~~~~~~~~~~~~~~~~~~~~~~~~~~~~~~~~~~~~~~~~~~~~~~~~~~~~\qed
\end{pf}
\noindent
The main challenge in imposing \eqref{eq: set inclusion} is that the probabilistically safe sets $\bar{\mathcal{X}}^{\text{safe}}_{t|\tau+1}$ are not yet determined at planning step $\tau$.
With this perspective, our goal is to design a \textit{probablistic recursively feasible} planner, named PRF-MPC, which satisfies the sufficient condition \eqref{eq: set inclusion} with a specified violation tolerance $\gamma$, i.e.,
\begin{equation}    
    \mathbb{P}\Bigg(\bigwedge_{t=2}^T\bigwedge_{\tau=0}^{t-1} \bar{\mathcal{X}}^{\text{safe}}_{t|\tau} \subseteq \bar{\mathcal{X}}^{\text{safe}}_{t|\tau+1}\Bigg) \geq 1-\gamma. \label{eq: recur chance constraint}
\end{equation}
To construct such a planner, 
we model the random set of the predicted $\bar{\mathcal{X}}^{\text{safe}}_{t|\tau+1}$ using only the information available at planning step $\tau$,
and denote this random set as $\widehat{\mathcal{X}}^{\text{safe}}_{t|\tau+1}$.
We will discuss this in detail in Section~\ref{sec:future_distribution}.



\subsection{Assumptions on Trajectory Predictions}
To formulate the safe sets in \eqref{eq: recur chance constraint}, we introduce Gaussian assumptions of the OV movements. Let $O_t\in \mathbb{R}^{n_o}$ denote the state (e.g., 2-dimensional position) of the OV at timestep $t$, and $O_{t|\tau}$ its predicted state at planning step $\tau$. 
\begin{assumption}\label{assumption:Gaussian} 
The state of the OV at any $t\in \mathbb{Z}_{1:T}$ predicted at $\tau\in \mathbb{Z}_{0:t-1}$ is Gaussian, i.e.,
\begin{equation}
  O_{t|\tau} \sim \mathcal{N}\!\big(\mu_{t|\tau},\, \Sigma_{t|\tau}\big).\nonumber
\end{equation}
\end{assumption}

\begin{assumption}\label{assumption:jointGaussian}
    For any $t_1, t_2 \in \mathbb{Z}_{\tau+1:T}$, the joint distribution $\big(O_{t_1|\tau},\, O_{t_2|\tau}\big)$ is also Gaussian.
\end{assumption}
\textit{Assumption~\ref{assumption:jointGaussian}} is also commonly adopted, as seen in Gaussian process–based dynamics models \citep{hewing2020simulation}.



\begin{figure}[t]
    \centerline{\includegraphics[width=9.0cm]{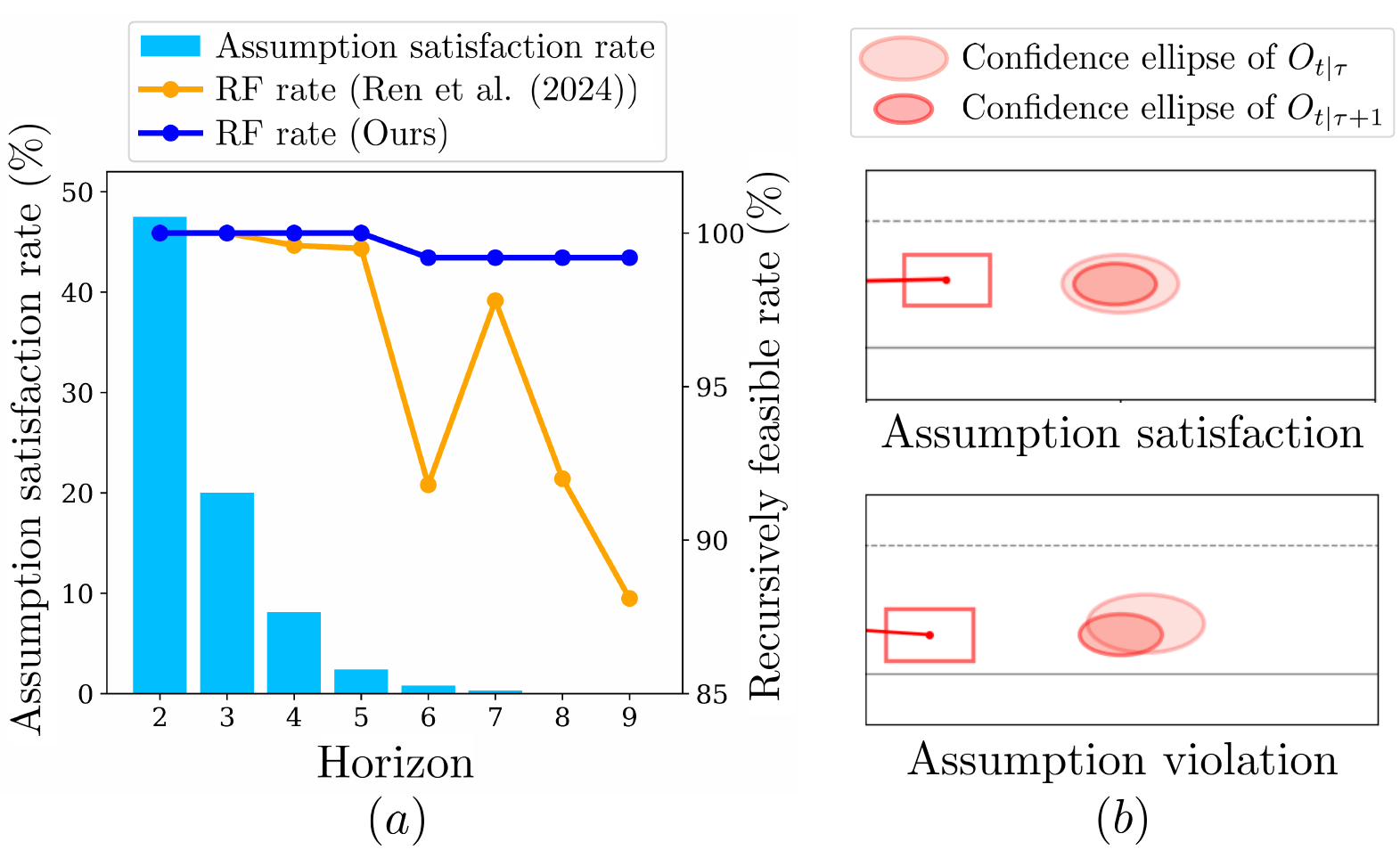}}
    \vspace*{-0.15in}
    \caption{(a) Satisfaction rate of \eqref{cond: covBoundmean} and recursively feasible rates with respect to horizon length $T$. (b) Examples of satisfaction and violation cases of \eqref{cond: covBoundmean}.}
    \label{fig: TCST assumption}
\end{figure}

\subsection{Existing Sufficient Condition and Its Limitation}
\cite{ren2024recursively} proposes a sufficient condition to satisfy the recursively feasible condition \eqref{eq: set inclusion}: for any  $t \in \mathbb{Z}_{2:T}$ and planning steps $\tau \in \mathbb{Z}_{0:t-1}$, the shift of the mean prediction is bounded by the covariance shift as follows:
    \begin{equation} \label{cond: covBoundmean}
         \norm{\mu_{t|\tau} - \mu_{t|\tau+1}}_2 \leq \Gamma_t \cdot \Big( \sqrt{\norm{\Sigma_{t|\tau}}_F} - \sqrt{\norm{\Sigma_{t|\tau+1}}_F} \Big),
    \end{equation}
where $\Gamma_t:=\Psi^{-1}(1-\epsilon_t)$ is the $(1-\epsilon_t)$ quantile of the standard Gaussian distribution.
However, this sufficient condition~\eqref{cond: covBoundmean} is rarely satisfied in practice even when using an \textit{ideal predictor} that follows reasonable properties, as discussed in Section~\ref{subsection:ideal predictor}. Using this ideal predictor, we evaluate the satisfaction rate of \eqref{cond: covBoundmean} and the resulting recursive feasibility rates (RF rates). As shown in Fig.~\ref{fig: TCST assumption}, the satisfaction rate decreases sharply as the planning horizon $T$ increases, reaching 0\% at $T=9$. This occurs because \eqref{cond: covBoundmean} must hold at every timestep across the entire horizon. Consequently, violations of this assumption cause a marked decline in the RF rate, whereas our proposed method maintains a high RF rate.






\section{Future Distribution Formulation}\label{sec:future_distribution}
The main challenge in ensuring recursive feasibility is that the future safe sets $\bar{\mathcal{X}}^{\text{safe}}_{t|\tau+1}$ in \eqref{eq: set inclusion} are unknown at planning step $\tau$. To model these future safe sets, $\widehat{\mathcal{X}}^{\text{safe}}_{t|\tau+1}$, we introduce two reasonable properties that a trajectory predictor should satisfy if it is to be considered ideal.

\subsection{Properties of an Ideal Predictor}\label{subsection:ideal predictor}

The first property is that an ideal predictor can sample from the true distribution of OV movements.

\textit{Property 1} (True Distribution Consistency).
Trajectory samples $(o_{\tau+1}, o_{\tau+2}, \dots, o_{\tau+T})^{(i)}, i \in \mathbb{Z}_{1:N}$ generated by an ideal predictor follow the true underlying distributions.

The second property captures the closed-loop behavior expected of an ideal predictor.

\textit{Property 2} (Conditional Invariance).
Conditional distributions are identical if conditioned on the same state and timestep:
\begin{align}
        \mathbb{P}(O_{t_2|\tau_1} \mid O_{t_1|\tau_0}=o)
    = \mathbb{P}(O_{t_2|\tau_2} \mid O_{t_1|\tau_0}=o), \nonumber \\
    \, \forall\, \tau_0\leq\tau_1,\tau_2 \le t_1 \le t_2, \nonumber
\end{align}
where $o$ is a realization of $O$.

For example, when $\tau_0=\tau_1=\tau$, $\tau_2=\tau+1$, and $t_1=\tau+1$, 
\[
\mathbb{P}(O_{t_2|\tau} \mid O_{\tau+1|\tau}=o)
= \mathbb{P}(O_{t_2|\tau+1} \mid O_{\tau+1|\tau}=o).
\]
This means that trajectories predicted at $\tau$ passing through $O_{\tau+1}=o$ follow the same conditional distribution as those predicted at $\tau+1$ starting from $O_{\tau+1}=o$.

\textit{Property 2} implies that once conditioned on the same state, future distributions remain invariant to the prediction time. This property is also desirable from a planning perspective, as it mitigates the temporal inconsistency of safe sets that can cause unstable planning behaviors in closed-loop planning.

\subsection{Deriving Future Predicted Distributions}\label{subsection:predict prediction}
We now derive how to estimate the distribution that will be predicted at future planning steps using information available at $\tau$ and the properties of an ideal predictor.

Specifically, we aim to compute the mean and covariance of future estimation $O_{t|\tau+a}$ for $1\le a < t-\tau$. 
At planning step $\tau$, $O_{\tau+a}$ is not yet determined. We therefore consider the conditional distribution $O_{t|\tau+a}$ given $O_{\tau+a|\tau} = o$, and obtain the overall distribution by marginalizing over the random variable $O_{\tau+a|\tau}$. The conditional mean is
\begin{align*}
    \mathbb{E}\big[O_{t|\tau+a} | O_{\tau + a|\tau}=o\big] 
    & = \mathbb{E}[O_{t|\tau}|O_{\tau+a|\tau}=o].
\end{align*}
This equality holds by \textit{Property~2} of the ideal predictor.
Therefore, the conditional mean $\mathbb{E}[O_{t|\tau+a} | O_{\tau + a|\tau}=o]$ can be expressed using only the information available at the current planning step $\tau$. By a similar derivation, the covariances satisfy an analogous relation: $Cov(O_{t|\tau+a}|O_{\tau+a|\tau})=Cov(O_{t|\tau}|O_{\tau+a|\tau})$.

Because the distribution of $O_{t|\tau}$ is known for all $t\in \mathbb{Z}_{\tau+1:T}$, i.e., $O_{t|\tau}\sim \mathcal{N}(\mu_{t|\tau}, \Sigma_{t|\tau})$ by \textit{Assumption~\ref{assumption:Gaussian}}, we can derive the conditional mean and covariance of $O_{t|\tau}$ given the Gaussian joint distribution assumption (\textit{Assumption~\ref{assumption:jointGaussian}}) via the well-known projection theorem,
\begin{align}
    &\mathbb{E}[O_{t|\tau}|O_{\tau+a|\tau}] = \notag \\
    &~~~~~\mu_{t|\tau}+\Sigma_{(t|\tau)(\tau+a|\tau)} (\Sigma_{\tau+a|\tau})^{-1}(O_{\tau+a|\tau}  - \mu_{\tau+a|\tau}),\label{eq: mu} \\
    &Cov(O_{t|\tau}|O_{\tau+a|\tau}) =\notag\\
    &~~~~~
    \Sigma_{t|\tau}-\Sigma_{(t|\tau)(\tau+a|\tau)} (\Sigma_{\tau+a|\tau})^{-1}\Sigma_{(\tau+a|\tau)(t|\tau)},
    \label{eq: cov}
\end{align}
where $\Sigma_{(t|\tau)(\tau+a|\tau)}$ is the cross covariance between $O_{t|\tau}$ and $O_{\tau+a|\tau}$.
Note that \eqref{eq: mu} is a random variable because it contains a random variable $O_{\tau+a|\tau}$ and \eqref{eq: cov} is a deterministic value.
Finally, the random variable $\mathbb{E}[O_{t|\tau+a}|O_{\tau+a|\tau}]$ follows
\begin{align}
\begin{split}\label{eq: mu_dist}
&\mathbb{E}[O_{t|\tau+a}|O_{\tau+a|\tau}] \\
&\sim\mathcal{N}\!\left(
\mu_{t|\tau},\,
\Sigma_{(t|\tau)(\tau+a|\tau)}(\Sigma_{\tau+a|\tau})^{-1}
\Sigma_{(\tau+a|\tau)(t|\tau)}\right).
\end{split}
\end{align}

In summary, we have derived the mean and covariance of prediction distributions at step $\tau+a$ based on the information available only at the current step $\tau$. Through this future information, we can construct the estimated random sets $\widehat{\mathcal{X}}^\text{safe}_{t|\tau+1}$.

\section{Recursively Feasible Planning}

\begin{figure}[t]
    \centerline{\includegraphics[width=6.0cm]{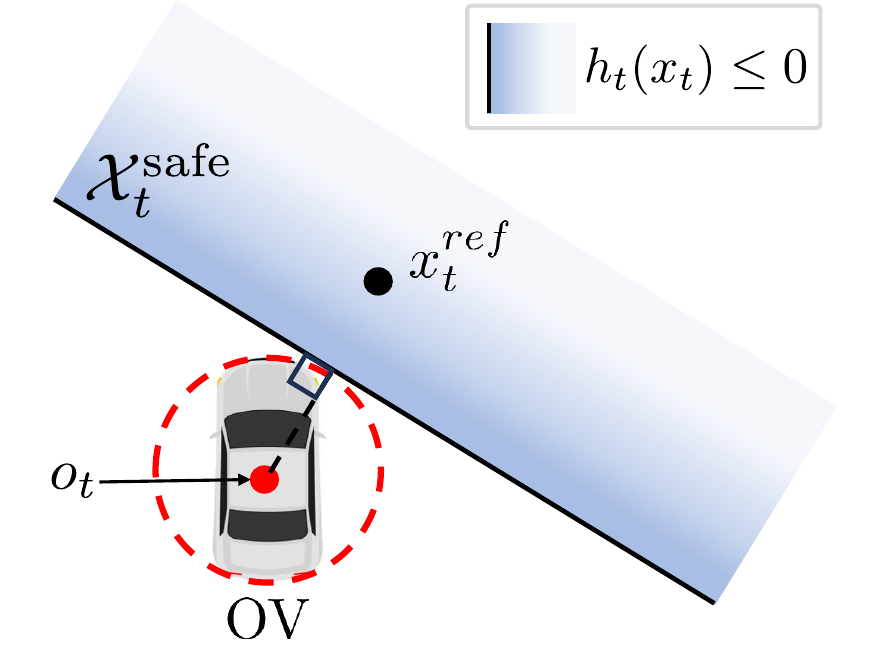}}
    \caption{Circular OV modeling and affine safety constraint.}
    \label{fig:affine constraint}
\end{figure}
In this section, we introduce a nominal MPC, and based on this, propose a PRF-MPC.

\subsection{Nominal MPC}
We model the OV as a circle in a 2-dimensional plane, allowing us to neglect heading uncertainty. 
The safe set of the EV with respect to the OV at time $t$ is defined as
\begin{equation}\label{eq:original_safety}
\{x_t \in \mathbb{R}^{n_x} \mid \|x_t - O_t\|_2^2 \ge r^2 \},
\end{equation}
where $r>0$ is the minimum safety distance.  

To obtain a computationally tractable approximation of \eqref{eq:original_safety}, we make use of a known reference trajectory $x_t^{\text{ref}}$ which remains independent of the actions of the MPC.
Let $h_t$ be the tangent line of the circle 
as shown in Fig.~\ref{fig:affine constraint}. 
The tangent line can be written as
\begin{align}
    h_t &:= -(x^{\text{ref}}_t-\mathbb{E}[O_t])^\intercal\cdot
    (x_t-O_t)
    + r \|(x^{\text{ref}}_t-\mathbb{E}[O_t])\|_2 \nonumber\\
    & = {\textbf{m}_t}^\intercal\big( x_t - O_t \big)+ r \|\textbf{m}_t\|_2=0, \nonumber
\end{align} 
where $\textbf{m}_t := -( x^{\text{ref}}_t-\mathbb{E}[O_t])$ denotes the vector from $x^{\text{ref}}_t$ to $\mathbb{E}[O_t]$, which determines the slope of the tangent line. 

We choose $\textbf{m}_t = -(x^{\text{ref}}_t-\mathbb{E}[O_{t|0}])$ from initial planning step $\tau=0$, and we omit the time index $t$ from $\textbf{m}_t$ in the remainder of the derivation for brevity. The affine safe set can be written as
\[
\mathcal{X}^{\text{safe}}_{t}
= \{x_t \in \mathbb{R}^{n_x} \mid
h_t(x_t) \le 0\},
\]
which conservatively approximates the original circular constraint \eqref{eq:original_safety}. Under this affine formulation, the joint chance-constrained safety constraint becomes
\begin{equation}    \mathbb{P}\Bigg(\bigwedge_{t=1}^T x_t\in\mathcal{X}^{\text{safe}}_{t}\Bigg) \geq 1-\epsilon. \label{eq: chance constraint}
\end{equation}
Each chance constraint can be reformulated in a deterministic form because the random variable $O_t$ follows a Gaussian distribution by \textit{Assumption~\ref{assumption:Gaussian}}. The deterministic form is 
\begin{align}\label{eq: nominal constraint0}
&\bigwedge_{t=1}^T  \textbf{m}^\intercal(x_t - \mu_t)  + r \|\textbf{m}\|_2 + \Gamma_t\sqrt{\textbf{m}^\intercal\Sigma_{t}\textbf{m}}  \leq 0,
\end{align}
where $\epsilon_{t}$ is a uniform allocation of $\epsilon$ over the planning horizon, i.e., $\epsilon_t= \epsilon/T$.
This is a conservative reformulation of the chance constraint, i.e. $ \eqref{eq: nominal constraint0}\Rightarrow \eqref{eq: chance constraint}$ \citep{lefkopoulos2021trajectory}.

With $l_{t|\tau}\big(x_{t|\tau}\big):=\textbf{m}^\intercal(x_{t|\tau} - \mu_{t|\tau}) + r\|\textbf{m}\|_2 + \Gamma_{t}\sqrt{\textbf{m}^\intercal\Sigma_{t|\tau}\textbf{m}},$ we  formulate the nominal MPC problem at time $\tau$ as follows:
\begin{subequations} \label{eq: nominal MPC}
    \begin{alignat} {2}
        & \underset{u_{\tau:T-1}}{\text{min}} &  &\mathcal{J}(x_{\tau+1:T}, u_{\tau:T-1})
        \\
        &\;\;s.t. & \quad &l_{t|\tau}(x_{t|\tau})\leq 0, \;\;
        \forall t \in \mathbb{Z}_{\tau+1:T}, \label{eq: nominal constraint}
        \\
        & & \; & (\ref{eq: dynamics}),(\ref{eq: x bounds}), (\ref{eq: u bounds}). \nonumber
    \end{alignat}
\end{subequations}
Note that \eqref{eq: nominal constraint} can only ensure the collision avoidance chance constraints~\eqref{eq: chance constraint}, but not the probabilistic recursively feasible condition~\eqref{eq: recur chance constraint} because it only considers single planning step constraints. The probabilistically safe sets from \eqref{eq: nominal constraint} are denoted as $\bar{\mathcal{X}}^\text{safe}_{t|\tau}$.


\subsection{PRF-MPC}
Instead of imposing \eqref{eq: nominal constraint}, we propose to enforce an alternative constraint $g_{t|\tau}(l_{t|\tau}(x_{t|\tau}))\leq 0$ where $g_{t|\tau}$ is a surrogate function that we will specify shortly. 
Based on this, the corresponding probabilistic safe sets are
\begin{equation*}
\widetilde{\mathcal{X}}^{\text{safe}}_{t|\tau} =\{x_t\in \mathbb{R}^{n_x} \mid g_{t|\tau}(l_{t|\tau}(x_t)) \leq 0\}.
\end{equation*}

Our goal is to construct $g_{t|\tau}$ that satisfies the following two conditions:
\begin{subequations}\label{eq:conditions}
\begin{enumerate}[label=(\theequation\alph*)]
    \item $g_{t|\tau}\big(l_{t|\tau}(x_{t|\tau})\big) \leq 0 \implies l_{t|\tau}(x_{t|\tau}) \leq 0$ \label{eq:cond_safe}
    \item $\mathbb{P} \Big(\widetilde{\mathcal{X}}^{\text{safe}}_{t|\tau} \subseteq 
    \widehat{\mathcal{X}}^{\text{safe}}_{t|\tau+1}\Big) \geq 1- \bar{\gamma}$ where $\bar{\gamma}=2\gamma/((T-1)T)$ for all $t \in \mathbb{Z}_{2:T}$ and $\tau\in\mathbb{Z}_{0:t-1}$.\label{eq:cond_RF}
\end{enumerate}
\end{subequations}
Condition~\ref{eq:cond_safe} ensures the original safety constraint \eqref{eq: nominal constraint}, and condition \ref{eq:cond_RF} ensures the probabilistic recursive feasibility condition \eqref{eq: recur chance constraint} by the following lemma. 

\begin{lem}\label{lem: Booles}
    Satisfaction of condition~\ref{eq:cond_RF} is sufficient to ensure \eqref{eq: recur chance constraint}.
\end{lem}
The proof is provided in Appendix~\ref{app: pf booles}..

In the following lemma, we provide an affine function $g_{t|\tau}$ that satisfies the two conditions.
\begin{lem}
    Define a positive constant $c^t_{\tau+1|\tau}$ as
    $$\text{max}\{ 
    - \Gamma_t \Big( \sqrt{\textbf{m}^\intercal\Sigma_{t|\tau}\textbf{m}}- \sqrt{ \textbf{m}^\intercal\hat{\Sigma}\textbf{m}} \Big) \notag
     +\Gamma_{\bar{\gamma}} \sqrt{\textbf{m}^\intercal\Sigma_{\hat{\mu}}\textbf{m}}, 0\},$$ where $\hat{\mu}:=\mathbb{E}[O_{t|\tau+1}|O_{\tau+1|\tau}], \Sigma_{\hat{\mu}}:=Cov\big(\hat{\mu}\big),$  $\hat{\Sigma}:=Cov[O_{t|\tau+1}|O_{\tau+1|\tau}]$, and $\Gamma_{\bar{\gamma}}:=\Psi^{-1}(1-\bar\gamma)$. Then, an affine function
    $$g_{t|\tau}\big(l_{t|\tau}(x_{t|\tau})\big)= l_{t|\tau}(x_{t|\tau})+\underset{i=\tau}{\overset{t-2}{\sum}}c^t_{i+1|i}$$ satisfies conditions \ref{eq:cond_safe} and \ref{eq:cond_RF}. \label{lem: RecurFeasi Const}
\end{lem}
The proof is provided in Appendix~\ref{app: pf RF const}. The proof of \textit{Lemma~\ref{lem: RecurFeasi Const}} is based on the construction of $\widehat{\mathcal{X}}^{\text{safe}}_{t|\tau+1}$ using information available only at the current step $\tau$, following the discussion of future trajectory distributions in Section~\ref{sec:future_distribution}. 

With \textit{Lemma} \textit{\ref{lem: RecurFeasi Const}}, we formulate the PRF-MPC problem as follows.
\begin{subequations} \label{eq: PRF-MPC}
    \begin{alignat} {2}
        & \underset{u_{\tau:T-1}}{\text{min}} &  &\mathcal{J}(x_{(\tau+1):T}, u_{\tau:T-1})
        \\
        &\;\;s.t. && 
        l_{t|\tau}(x_{t|\tau})+\underset{i=\tau}{\overset{t-2}{\sum}}c^t_{i+1|i}
        \leq 0,          \forall t \in \mathbb{Z}_{(\tau+1):T}, \label{eq: recur constraint}
        \\
        & & \quad & (\ref{eq: dynamics}),(\ref{eq: x bounds}),(\ref{eq: u bounds}). \nonumber
    \end{alignat}
\end{subequations}
In the following, we prove that PRF-MPC provides probabilistic safety~\eqref{eq: chance constraint} and probabilistic recursive feasibility guarantee~\eqref{eq: recur chance constraint}.
\setcounter{thm}{0}
\begin{thm} Suppose \textit{Assumptions}\textit{~\ref{assumption:Gaussian}} and \textit{\ref{assumption:jointGaussian}} hold, and our trajectory predictor has \textit{Properties} \textit{1} and \textit{2}. When a feasible solution exists at $\tau=0$, the PRF-MPC remains feasible for all $\tau = 1,\ldots,T-1$, with probability at least $1-\gamma$, and the planned trajectory at each planning step ensures the chance constraint \eqref{eq: chance constraint}.
\end{thm}

\begin{pf}
We show that satisfying constraint \eqref{eq: recur constraint} implies satisfaction of both \eqref{eq: recur chance constraint} and \eqref{eq: chance constraint}. By \textit{Lemma~\ref{lem: RecurFeasi Const}}, the constraint \eqref{eq: recur constraint} ensures condition \ref{eq:cond_RF}, which in turn implies the probabilistic recursively feasible condition \eqref{eq: recur chance constraint} by \textit{Lemma~\ref{lem: Booles}}. Furthermore, \textit{Lemma~\ref{lem: RecurFeasi Const}} also guarantees condition \eqref{eq: nominal constraint}, which is equivalent to \eqref{eq: nominal constraint0}. The constraint \eqref{eq: nominal constraint0} then ensures the chance constraint \eqref{eq: chance constraint}~\citep{lefkopoulos2021trajectory}. ~~~~~~~~~~~~~~~~~~~~~~~~~~~~~~~~~~~~~~~~~~~\qed
\end{pf}

\textit{Remark 1.} The probabilistic recursive feasibility guarantee developed for unimodal Gaussian prediction can be directly extended to the Gaussian Mixture Model (GMM) case. As in \cite{ren2024recursively}, we impose the same constraint \eqref{eq: recur constraint} on each mixture component and uniformly allocate the violation tolerance $\gamma$ across the modes. Under the assumption that the active modes remain unchanged or decrease over the shrinking horizon, the condition \eqref{eq: recur chance constraint} continues to hold. Moreover, this framework naturally generalizes to scenarios with multiple OVs, since the joint safety constraints can be enforced independently for each OV using its own predicted distribution.

\section{Case Study}\label{sec: case study}

In this section, we conduct simulations of a lane-change task and validate the recursive feasibility satisfaction of our proposed method.  
Experiments are conducted with Intel Core i7-14700K (3.40 GHz) to run the simulations, and optimizers using CVXPY and Gurobi 13.0.0.

\subsection{Experiment Settings}

The EV aims to merge into an adjacent lane occupied by an OV. 
The EV dynamics is modeled as a double-integrator
\begin{equation}
 \dot{x} = (\dot{p_1}, \dot{p_2}, \dot{v_1}, \dot{v_2})^\top = (v_1,\, v_2,\, u_1,\, u_2)^\top, \nonumber
\end{equation}
where $(p_1, p_2)$ denote the planar position and $(v_1, v_2) \in [0, 30]\times[-5, 5]\text{m/s}$ denote the linear velocity, and the control inputs $(u_1, u_2)\in [-10, 10]\times[-5, 5]\text{m/s}^2$ represent the linear accelerations. 
For the OV, we use a single-integrator model
\[
\dot{o} = (\dot{p}_1^{\mathrm{OV}},\, \dot{p}_2^{\mathrm{OV}})^\top 
       = (v_1^{\mathrm{OV}},\, v_2^{\mathrm{OV}})^\top,
\]
where \((p_1^{\mathrm{OV}}, p_2^{\mathrm{OV}})\) denote the planar position of OV.  
To generate the OV's trajectory distribution, we sample its velocity at each timestep as
\begin{equation}
(v_1^{\mathrm{OV}},\, v_2^{\mathrm{OV}}) 
\sim \mathcal{N}\!\left(
[15,\,0],\;\nonumber
diag(1.0, 0.25)
\right).
\end{equation}
These continuous models are discretized with a sampling time of 0.5\,s using the forward Euler discretization. The objective function is defined as the 2-norm of the deviation from a reference trajectory, guiding the EV along a prescribed lane-change maneuver. We set the safe distance $r=4$, the planning horizon $T = 9$, a risk tolerance $\epsilon = 0.05$, and a recursive feasibility violation tolerance $\gamma=0.1$. 

We conduct 1,000 trials and report the average results.
As for the metrics, we use the percentage of trials in which the planning remains recursively feasible (``RF Rate''), the 2-norm of the deviation of closed-loop trajectories from the reference trajectory (``Cost''), the minimum distance from the OV (``$\text{d}_{\min}$''), and the average worst-case solve time (``Comp. Time'').

\subsection{Results}
The results of the lane changing experiments are illustrated in Fig.~\ref{fig:traj}, Fig.~\ref{fig: check constraints},  and Table~\ref{tb: results}.
\begin{figure}[h!]
    \centering

    \centerline{\includegraphics[width=9.0cm]{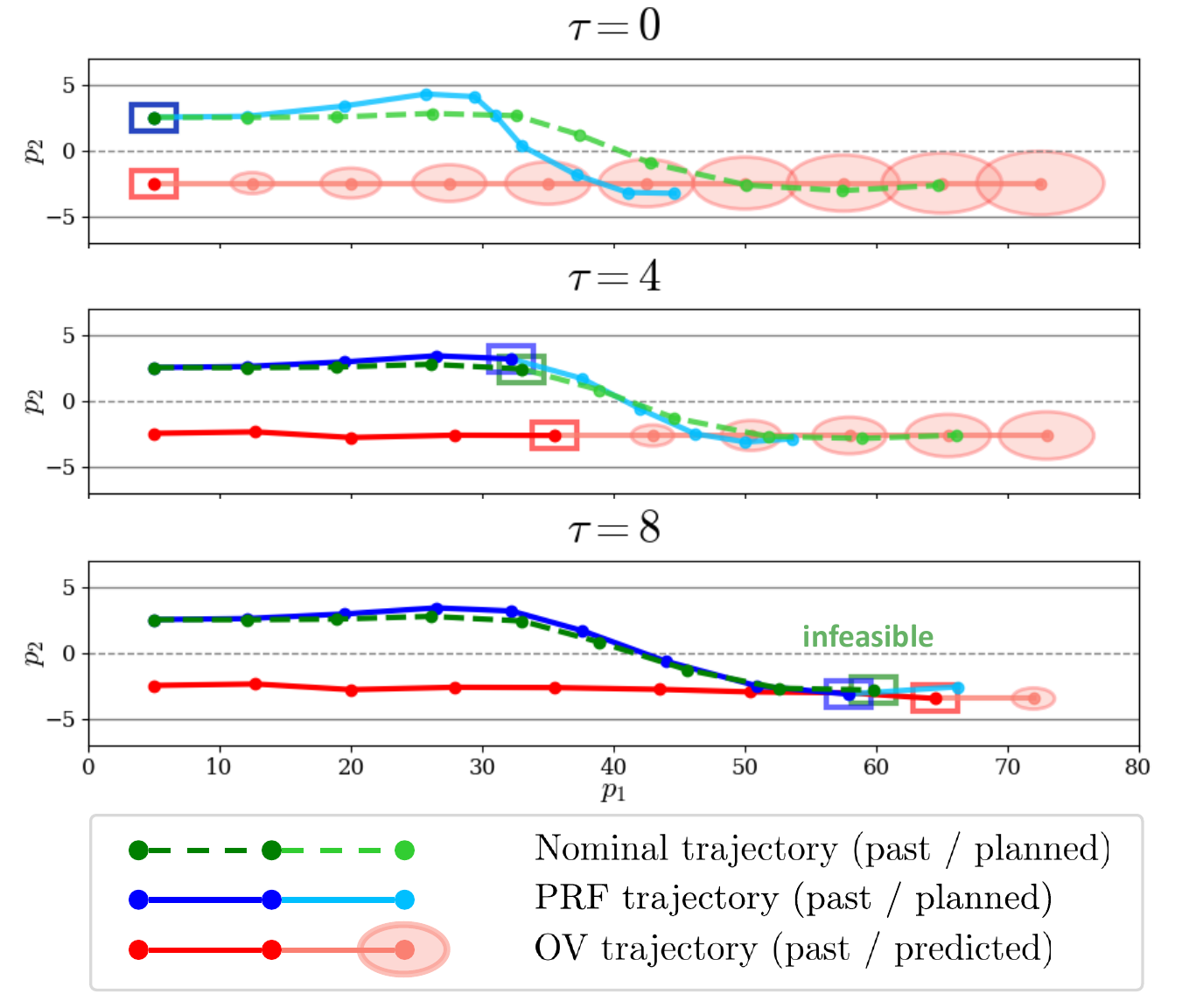}}
    \caption{
    The nominal MPC fails to retain feasibility during closed-loop execution, while PRF-MPC guarantees recursive feasibility until the goal is reached.
    }
    \label{fig:traj}
\end{figure}

\textit{Recursive Feasibility:} As shown in Table~\ref{tb: results}, PRF-MPC achieves a 99.2\% RF rate, whereas nominal MPC results in less than 90\%. This confirms that PRF-MPC ensures the probability of recursive feasibility greater than $1-\gamma$, and suggests that the actual probability of maintaining recursive feasibility is significantly higher than the tolerance $1-\gamma$.

Fig.~\ref{fig:traj} illustrates a scenario in which the nominal MPC becomes infeasible while PRF-MPC remains feasible. The corresponding safe sets at the planning steps are shown in Fig.~\ref{fig: check constraints}. For the nominal MPC, the probabilistically safe set $\bar{\mathcal{X}}_{9|7}^{\text{safe}}$ is not contained within $\bar{\mathcal{X}}_{9|8}^{\text{safe}}$, leading to a violation of recursive feasibility. In contrast, PRF-MPC preserves the set inclusion, i.e., $\widetilde{\mathcal{X}}_{9|7}^{\text{safe}}\subseteq \widetilde{\mathcal{X}}_{9|8}^{\text{safe}}$, thereby maintaining feasibility.

\textit{Conservatism and Safety:} Table~\ref{tb: results} shows that PRF-MPC yields a higher cost and a larger $\text{d}_{\text{min}}$, indicating that it produces more conservative yet safer trajectories. As shown in Fig.~\ref{fig:traj}, although PRF-MPC’s step-wise plans are more conservative, this conservativeness diminishes over time through closed-loop execution.

\textit{Computation Time:} In Table~\ref{tb: results}, the computation time of 0.034\,s demonstrates that the proposed PRF-MPC is suitable for real-time implementation. This efficiency is primarily enabled by the tangent-line approximation of the original nonconvex collision-avoidance constraint in \eqref{eq: chance constraint}. Moreover, PRF-MPC does not incur any additional computational cost relative to the nominal MPC.
 
In summary, our proposed method significantly improves recursive feasibility over nominal MPC approaches with low computational costs. Moreover, the conservativeness of PRF-MPC can be progressively reduced during closed-loop execution. 


\begin{table}[t]
\begin{center}
\caption{Simulation results}\label{tb: results}
\begin{tabular}{|c|c|c|c|c|}
\hline
Planning & RF Rate (\%) & Cost & $\text{d}_{\text{min}}$ & Comp. Time (s) \\\hline
nominal & 88.2 & 25.15 & 4.74 & 0.034\\ 
PRF & 99.2 & 69.38 & 4.95 & 0.034\\ \hline
\end{tabular}
\end{center}
\end{table}

\begin{figure}[t]
    \centerline{\includegraphics[width=9.0cm]{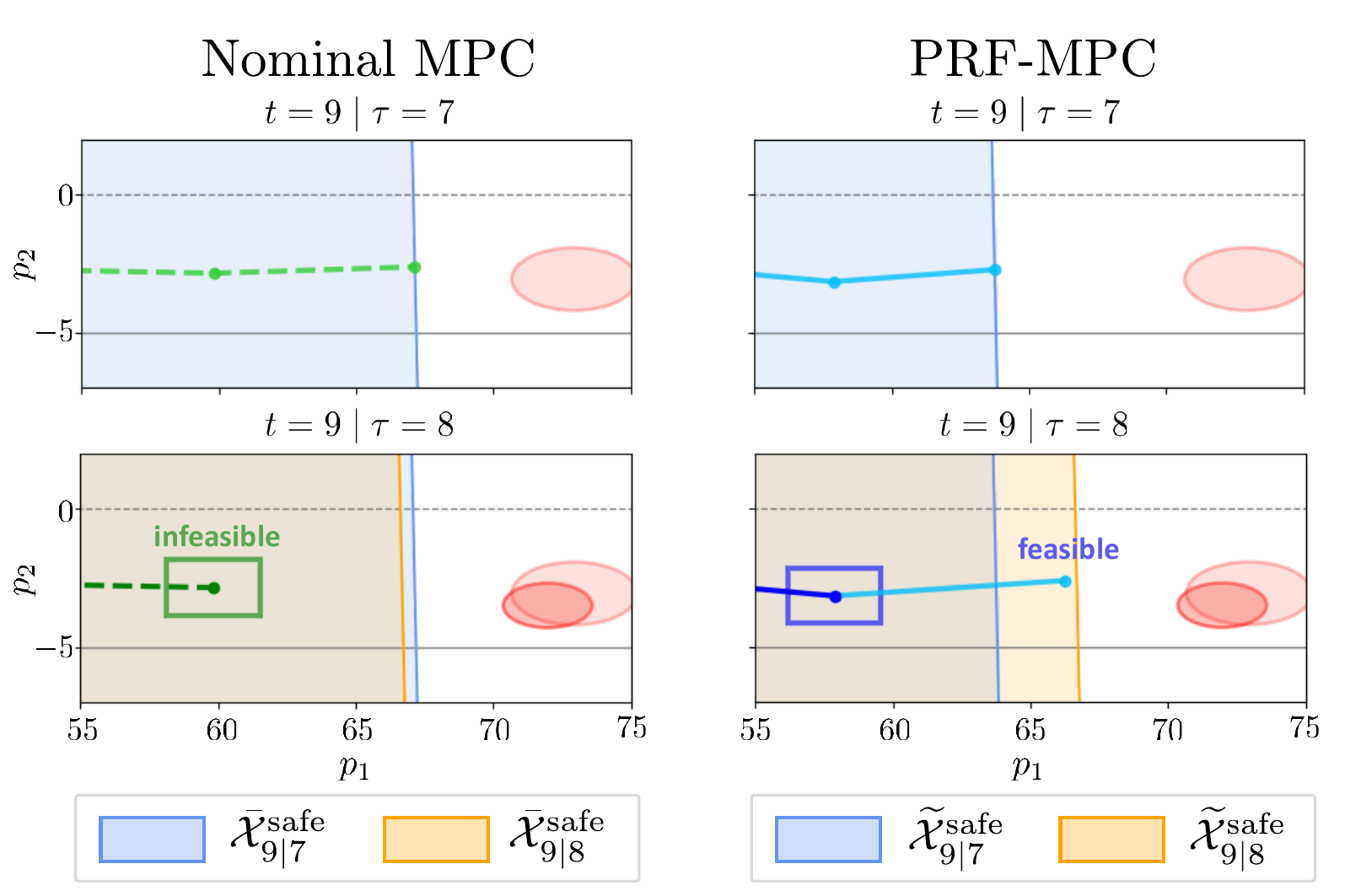}}
    \caption{Safe set inclusion and recursive feasibility.}
    \label{fig: check constraints}
\end{figure}

\section{Conclusion}
We have proposed a probabilistic recursively feasible chance-constrained MPC framework under environmental uncertainty.  Along the way, we have proposed the key properties of an ideal predictor, and then formulated the prediction propagation. This formulation enables constructing the safe set inclusion constraint for the MPC framework, ensuring probabilistic recursive feasibility. We validated our approach in an autonomous-driving lane-change scenario, demonstrating a significant improvement in recursive feasibility.
Our method can be extended to more complex scenarios, such as multiple obstacles and multiple prediction modes.
Nevertheless, current off-the-shelf predictors do not fully possess the properties of the ideal predictor yet, leaving practical approximations and their impact as future work.


\begin{ack}
We sincerely appreciate Professor Maryam Kamgarpour at EPFL for inspiring this research direction and for the valuable discussions during our visit via the EPFL E3 program.
\end{ack}


\bibliography{ifacconf}             
\appendix

\section{Proof}    
\subsection{Moments of $\mathbb{E}[O_{t|\tau+a}|O_{\tau+a|\tau}]$} \label{app: mu_dist}
The derivation of $\eqref{eq: mu_dist}$ follows,
\begin{subequations}
\begin{align}
    &\mathbb{E}\Big[\mathbb{E}[O_{t|\tau+a}|O_{\tau+a|\tau}=o]\Big] \notag \\ 
    &= \mathbb{E}\Big[\mathbb{E}[O_{t|\tau}]+\Sigma_{(t|\tau)(\tau+a|\tau)}(\Sigma_{\tau+a|\tau})^{-1}(o - \mathbb{E}[O_{\tau+a|\tau}])\Big] \notag \\
    &= \mathbb{E}[O_{t|\tau}],\notag \\
    &Cov\Big[\mathbb{E}[O_{t|\tau+a}|O_{\tau+a|\tau}=o]\Big] \notag \\ 
    &= Cov\Big[\mathbb{E}[O_{t|\tau}]+\Sigma_{(t|\tau)(\tau+a|\tau)}(\Sigma_{\tau+a|\tau})^{-1}(o - \mathbb{E}[O_{\tau+a|\tau}])\Big] \notag \\
    &= Cov\Big[\Sigma_{(t|\tau)(\tau+a|\tau)}(\Sigma_{\tau+a|\tau})^{-1}o\Big] \notag \\
    &= \Sigma_{(t|\tau)(\tau+a|\tau)}(\Sigma_{\tau+a|\tau})^{-1}\Sigma_{(t|\tau)(\tau+a|\tau)}. \notag 
\end{align}
\end{subequations}
\subsection{Proof of \textit{Lemma}~\textit{\ref{lem: Booles}}} \label{app: pf booles}
We want to show that \ref{eq:cond_RF} implies \eqref{eq: recur chance constraint}. 
For brevity, let $A_{t|\tau}:= \widetilde{\mathcal{X}}^{\text{safe}}_{t|\tau} \subseteq \widehat{\mathcal{X}}^{\text{safe}}_{t|\tau+1}$. Then condition \ref{eq:cond_RF} can be expressed as $\mathbb{P}\big(A_{t|\tau}\big) \geq 1- \bar{\gamma}$, and this is equivalent to $\mathbb{P}\big(\neg A_{t|\tau}\big) \leq  \bar{\gamma}$. 
Summing over all $t \in \mathbb{Z}_{2:T}$ and $\tau\in\mathbb{Z}_{0:t-1}$, we obtain $\sum_{t=2}^{T} \sum_{\tau=0}^{t-1} \mathbb{P}\big(\neg A_{t|\tau}\big) \leq \gamma$ because $\bar\gamma(T-1)T/2=\gamma$.
By Boole's inequality, $\mathbb{P}\Big(\bigvee_{t=2}^T\bigvee_{\tau=0}^{t-1} \neg A_{t|\tau}\Big) \leq \gamma$.
By De Morgan's laws, this is equivalent to $1-\mathbb{P}\Big(\bigwedge_{t=2}^T\bigwedge_{\tau=0}^{t-1}A_{t|\tau}\Big) \leq \gamma $. 
Finally, this is equivalent to \eqref{eq: recur chance constraint}, because $\widetilde{\mathcal{X}}^{\text{safe}}_{t|\tau}$ is $\bar{\mathcal{X}}^{\text{safe}}_{t|\tau}$ of \eqref{eq: recur chance constraint}, and we treat $\bar{\mathcal{X}}^{\text{safe}}_{t|\tau+1}$ of \eqref{eq: recur chance constraint} as a realization of $\widehat{\mathcal{X}}^{\text{safe}}_{t|\tau+1}$.~~~~~~~~~~~~~~~~~~~~~~~~~~~~~~~~~~~~~~~~~~~~~~~~~~~~~~~~~~~~~~~~~~\qed

\subsection{Proof of \textit{Lemma}~\textit{\ref{lem: RecurFeasi Const}}} \label{app: pf RF const}

    We first want to show that $g_{t|\tau}\big(l_{t|\tau}(x_{t|\tau})\big)$ satisfies \ref{eq:cond_RF}. Note that $ \widehat{\mathcal{X}}^{\text{safe}}_{t|\tau+1}$ in \ref{eq:cond_RF} has a random variable, $\hat{\mu}=\mathbb{E}[O_{t|\tau+1}|O_{\tau+1|\tau}]$, and we define 
\begin{align*}
&\hat{l}_{t|\tau+1}\big(x_{t}\big):=\textbf{m}^\intercal(x_{t} - \hat{\mu}) + r\|\textbf{m}\|_2 + \Gamma_{t}\sqrt{\textbf{m}^\intercal\hat{\Sigma}\textbf{m}},
\\
&\widehat{\mathcal{X}}^{\text{safe}}_{t|\tau+1}:=\{{x_t\in\mathbb{R}^{n_x}\mid g_{t|\tau}\big(\hat{l}_{t|\tau}(x_{t})\big)\leq0}\}
\end{align*}
    where $ \Sigma_{\hat{\mu}}:=Cov\big(\hat{\mu}\big),$ and $\hat{\Sigma}:=Cov[O_{t|\tau+1}|O_{\tau+1|\tau}]$.
    Then, for any $x_t\in\mathbb{R}^{n_x}$, because we use $\textbf{m}_t = -(x^{\text{ref}}_t-\mathbb{E}[O_{t|0}]), \forall t$,
\begin{align*}
    &\widetilde{\mathcal{X}}^{\text{safe}}_{t|\tau} \subseteq 
    \widehat{\mathcal{X}}^{\text{safe}}_{t|\tau+1} \\
    &\iff \{   g_{t|\tau}\big(l_{t|\tau}(x_{t})\big) \leq0 \} \subseteq 
    \{ g_{t|\tau+1}\big(\hat{l}_{t|\tau+1}(x_{t})\big) \leq0 \} \\
    &\iff
    \{  g_{t|\tau}\big(l_{t|\tau}(x_{t})\big) \leq g_{t|\tau+1}\big(\hat{l}_{t|\tau+1}(x_{t})\big)  \leq 0 \}.
\end{align*}
Then, for any $x_t\in\mathbb{R}^{n_x}$, $g_{t|\tau}$ needs to satisfy,
\begin{align}
    \ref{eq:cond_RF}& \iff 
    \mathbb{P}\Big(g_{t|\tau}\big(l_{t|\tau}(x_{t})\big) \leq  g_{t|\tau+1}\big(\hat{l}_{t|\tau+1}(x_{t})\big) \Big) \geq 1-\bar{\gamma}\notag \\
    \iff &\mathbb{P}\Big( l_{t|\tau}(x_{t})+\underset{i=\tau}{\overset{t-2}{\sum}}c^t_{i+1|i} \leq \hat{l}_{t|\tau+1}(x_{t})+\underset{i=\tau+1}{\overset{t-2}{\sum}}c^t_{i+1|i} \Big) \geq 1 - \bar{\gamma}\notag\\
    \iff & \mathbb{P}\Big( l_{t|\tau}(x_{t}) \leq \hat{l}_{t|\tau+1}(x_{t})+c^t_{\tau+1|\tau} \Big)\geq 1 - \bar{\gamma}, \label{eq: recur feasi one step}
\end{align}
    where the constant $c^t_{\tau+1|\tau}$ is an increment margin $l_{t|\tau}(x_{t})$ and $\hat{l}_{t|\tau+1}(x_{t})$.
    
For brevity, we define $A:= \textbf{m}^\intercal\cdot\mu_{t|\tau+1} - \Gamma_t \Big( \sqrt{ \textbf{m}^\intercal\Sigma_{t|\tau}\textbf{m} } - \sqrt{ \textbf{m}^\intercal\hat{\Sigma}\textbf{m} } \Big)$, and to derive $c^t_{\tau+1|\tau}$, we reformulate \eqref{eq: recur feasi one step} as following,
\begin{align} \label{eq: recur feasi const}
    \eqref{eq: recur feasi one step} 
    &\iff  \mathbb{P} \Bigg( \textbf{m}^\intercal \cdot \hat{\mu} \geq  A -c^t_{\tau+1|\tau} \Bigg) \geq 1- \bar{\gamma}\notag\\
    \iff &\mathbb{P} \Bigg(  z\geq \frac{ A -c^t_{\tau+1|\tau}  - \mathbb{E}\big[\textbf{m} \cdot \hat{\mu} \big]}{\sqrt{\textbf{m}^\intercal\Sigma_{\hat{\mu}}\textbf{m}}} \Bigg)  \geq 1- \bar{\gamma} \notag\\
    \iff & A - c^t_{\tau+1|\tau}  - \textbf{m} \cdot \mu_{t|\tau+1} \leq \Phi^{-1}(\bar{\gamma})\sqrt{\textbf{m}^\intercal\Sigma_{\hat{\mu}}\textbf{m}} \notag
    \\
    \iff & c^t_{\tau+1|\tau} \geq  
    - \Gamma_t \Big( \sqrt{\textbf{m}^\intercal\Sigma_{t|\tau}\textbf{m}}- \sqrt{ \textbf{m}^\intercal\hat{\Sigma}\textbf{m}} \Big) \notag\\
    & +\Gamma_{\bar{\gamma}} \sqrt{\textbf{m}^\intercal\Sigma_{\hat{\mu}}\textbf{m}}  \nonumber
\end{align}
where the first equivalence follows from algebraic manipulation of $l_{t|\tau}(x_{t|\tau})$ and $\hat{l}_{t|\tau+1}(x_{t|\tau})$, the second from standardizing $\hat{\mu}$, and the third from reformulating the constraint deterministically using properties of Gaussian random variables. We choose a non-negative value $$c^t_{\tau+1|\tau} := \text{max}\{ 
    - \Gamma_t \Big( \sqrt{\textbf{m}^\intercal\Sigma_{t|\tau}\textbf{m}}- \sqrt{ \textbf{m}^\intercal\hat{\Sigma}\textbf{m}} \Big) \notag
     +\Gamma_{\bar{\gamma}} \sqrt{\textbf{m}^\intercal\Sigma_{\hat{\mu}}\textbf{m}}, 0\}.$$
 By construction of $c^t_{\tau+1|\tau}$, $g_{t|\tau}\big(l_{t|\tau}(x_{t})\big)$ satisfies~\ref{eq:cond_RF}. It also satisfies condition~\ref{eq:cond_safe} because $c^t_{\tau+1|\tau}\geq0$, $g\big(l_{t|\tau}(x_{t})\big) \leq0 \iff l_{t|\tau}(x_{t}) +C_{\tau+1|\tau} \leq0 \Rightarrow l_{t|\tau}(x_{t})\leq0.$~~~~~~~~~~~~~~~~~~~~~~~~~~~~~~~~~~~~~~~~~~~~~~~~~~~~~~~~~~~~~~~~~~~~~~~~~\qed

\end{document}

%% file: notation.tex
\newcommand{\bmat}{\begin{bmatrix}}
\newcommand{\emat}{\end{bmatrix}}

\newcommand{\bsmat}{\begin{bsmallmatrix}}
\newcommand{\esmat}{\end{bsmallmatrix}}

\providecommand{\norm}[1]{\left\lVert#1\right\rVert}

%% file: ifacconf.bib
@article{lefkopoulos2021trajectory,
  title={Trajectory planning under environmental uncertainty with finite-sample safety guarantees},
  author={Lefkopoulos, Vasileios and Kamgarpour, Maryam},
  journal={Automatica},
  volume={131},
  pages={109754},
  year={2021},
  publisher={Elsevier}
}

@article{ahn2021safe,
  title={Safe motion planning against multimodal distributions based on a scenario approach},
  author={Ahn, Heejin and Chen, Colin and Mitchell, Ian M and Kamgarpour, Maryam},
  journal={IEEE Control Systems Letters},
  volume={6},
  pages={1142--1147},
  year={2021},
  publisher={IEEE}
}

@article{ren2022chance,
  title={Chance-constrained trajectory planning with multimodal environmental uncertainty},
  author={Ren, Kai and Ahn, Heejin and Kamgarpour, Maryam},
  journal={IEEE Control Systems Letters},
  volume={7},
  pages={13--18},
  year={2022},
  publisher={IEEE}
}

@article{lindemann2023safe,
  title={Safe planning in dynamic environments using conformal prediction},
  author={Lindemann, Lars and Cleaveland, Matthew and Shim, Gihyun and Pappas, George J},
  journal={IEEE Robotics and Automation Letters},
  volume={8},
  number={8},
  pages={5116--5123},
  year={2023},
  publisher={IEEE}
}

@article{ren2024recursively,
  title={Recursively Feasible Chance-Constrained Model Predictive Control Under Gaussian Mixture Model Uncertainty},
  author={Ren, Kai and Chen, Colin and Sung, Hyeontae and Ahn, Heejin and Mitchell, Ian M and Kamgarpour, Maryam},
  journal={IEEE Transactions on Control Systems Technology},
  volume={33},
  number={4},
  pages={1193--1206},
  year={2024},
  publisher={IEEE}
}

@inproceedings{stamouli2024recursively,
  title={Recursively feasible shrinking-horizon MPC in dynamic environments with conformal prediction guarantees},
  author={Stamouli, Charis and Lindemann, Lars and Pappas, George},
  booktitle={6th Annual Learning for Dynamics \& Control Conference},
  pages={1330--1342},
  year={2024},
}

@inproceedings{ivanovic2019trajectron,
  title={The trajectron: Probabilistic multi-agent trajectory modeling with dynamic spatiotemporal graphs},
  author={Ivanovic, Boris and Pavone, Marco},
  booktitle={Proceedings of the IEEE/CVF international conference on computer vision},
  pages={2375--2384},
  year={2019}
}

@inproceedings{varadarajan2022multipath++,
  title={Multipath++: Efficient information fusion and trajectory aggregation for behavior prediction},
  author={Varadarajan, Balakrishnan and others},
  booktitle={2022 IEEE International Conference on Robotics and Automation (ICRA)},
  pages={7814--7821},
  year={2022},
}

@inproceedings{hewing2020simulation,
  title={On simulation and trajectory prediction with Gaussian process dynamics},
  author={Hewing, Lukas and Arcari, Elena and Fr{\"o}hlich, Lukas P and Zeilinger, Melanie N},
  booktitle={Learning for Dynamics and Control},
  pages={424--434},
  year={2020},
}

@inproceedings{chen2023invariant,
  title={Invariant safe contingency model predictive control for intersection coordination of mixed traffic},
  author={Chen, Xiao and M{\aa}rtensson, Jonas},
  booktitle={2023 IEEE 26th International Conference on Intelligent Transportation Systems (ITSC)},
  pages={3369--3376},
  year={2023},
  organization={}
}

@article{brudigam2021stochastic,
  title={Stochastic model predictive control with a safety guarantee for automated driving},
  author={Br{\"u}digam, Tim and Olbrich, Michael and Wollherr, Dirk and Leibold, Marion},
  journal={IEEE Transactions on Intelligent Vehicles},
  volume={8},
  number={1},
  pages={22--36},
  year={2021},
  publisher={IEEE}
}

@article{yang2026safe,
  title={Safe and Nonconservative Contingency Planning for Autonomous Vehicles via Online Learning-Based Reachable Set Barriers},
  author={Yang, Rui and Zheng, Lei and Ge, Shuzhi Sam and Ma, Jun},
  journal={IEEE Transactions on Control Systems Technology},
  year={2026},
  publisher={IEEE}
}
